\documentclass[12pt,a4paper,oneside]{article}
\usepackage[hidelinks]{hyperref}

\usepackage{style}
\newcommand{\ch}{\mathrm{ch}}

\theoremstyle{theorem}
\newtheorem{axprime}{Axiom}

\crefname{axprime}{axiom}{axioms}
\usepackage{multirow}

\usepackage{xfrac}
\AtEveryCitekey{\clearlist{location}}
\AtEveryBibitem{\clearlist{location}}
\usepackage[group-separator={,},group-minimum-digits={3}]{siunitx}

\begin{document}
\cleanlookdateon

\begin{center}
\textbf{Social choice with transfers}\footnote{We thank Chris Chambers, Jerry Green, Markus Möller, Hervé Moulin, William Thomson, and Peyton Young for valuable comments. Curello gratefully acknowledges support from the German Research Foundation (DFG) through CRC TR 224 (Project B02).}

\vspace{\baselineskip}
Gregorio Curello\footnote{Department of Economics, University of Mannheim; \texttt{gregorio.curello@uni-mannheim.de}.} and Sam Jindani\footnote{Department of Economics, National University of Singapore; \texttt{sam.jindani@nus.edu.sg}.}

\vspace{\baselineskip}
\today
\end{center}

\begin{quote}
	We consider the problem of social choice when transfers between agents are possible. This includes several canonical applications: public-good provision, management of a common resource, settlement of debts, and division of goods. The question of interest is, given a set of possible alternatives, which should be chosen and what transfers, if any, should be made? We show that the \emph{Shapley value of the stand-alone game} is the only solution to satisfy certain desirable properties. \\

\end{quote} 

\onehalfspacing

\section{Introduction}\label{sec_intro}

Consider the problem of social choice when agents have quasi-linear utility in a numéraire good, such as money, and transfers of the numéraire are possible. To fix ideas, suppose a government has earmarked funds to build a power plant. The power plant will serve two cities and there are two possible locations. City 1 derives a benefit of 8 in monetary units from location $A$ and a benefit of 2 from location $B$; city 2 derives benefits of 2 from $A$ and 4 from $B$. The problem is represented in \cref{fig_plant}. The dark-grey area is the set of allocations that are feasible without transfers (allowing for free disposal and randomisation). Assuming the government can make transfers between the cities, the set of feasible allocations is the light-grey area. What is a fair solution to this problem? It seems natural for the power plant to be built at location $A$, since the alternative would be Pareto inefficient. However, city 2 might argue that it should be compensated for forgoing its preferred alternative. Thus a fair solution might be for the plant to be built at location $A$ in exchange for city 1 paying compensation to city 2. The question, then, is how much.

\begin{figure}[t]
	\centering
	\begin{tikzpicture}[scale=.9]
		\draw[fill=gray!20] (0,10) -- (0,0) --  (10,0) -- cycle;
		\draw[fill=gray!50] (8,0) -- (8,2) -- (2,4) -- (0,4) --(0,0) -- cycle;
		\draw [<->] (0,10.5) node [above] {$x_2$} -- (0,0) -- (10.5,0) node [right] {$x_1$};
		\draw[fill] (8,2) circle [radius=.06] node[below left] {$(8,2)$};
		\draw[fill] (2,4) circle [radius=.06] node [below left] {$(2,4)$};
		\draw[fill] (7,3) circle [radius=.06] node [above right] {$x^*=(7,3)$};	
	\end{tikzpicture}
    \caption[]{Building a power plant.}
    \label{fig_plant}
\end{figure}

Beyond this example, we are interested in any situation in which a group of agents or a social planner must choose between multiple alternatives, preferences are quasi-linear in a numéraire good, and transfers of the numéraire are possible. The focus is on identifying a fair allocation assuming the preferences of the agents are known (à la \cite{arrow:1951,sen:1970}), and not on implementation (à la \cite{green:1979}). The data of the problem are a set of players and a set of allocations, each yielding a utility to each player. We will refer to this set of allocations as the \emph{bargaining set} and the set of utility profiles that can be achieved with transfers as the \emph{feasible set}. 

This problem was first formulated by \textcite{green:1983}; we argue that it can be seen as a generalisation of several canonical problems, including public-good provision, management of a common resource, settlement of debts, and division of goods. In each of these problems, a classical solution is what \textcite{moulin:1992} terms the \emph{Shapley value of the stand-alone game}.\footnote{The solution was proposed by \textcite{young:1982} for public-good provision; by \textcite{oneill:1982,aumann:1985a} for the settlement of debts; and by \textcite{moulin:1992} for the division of goods.} Given a social-choice problem, define the \emph{stand-alone claim} of a coalition of players to be the maximal total utility that the coalition can achieve in the bargaining set; this defines a coalitional game with transferable utility called the \emph{stand-alone game}; computing its Shapley value yields a solution to the social-choice problem. In the example of the power plant, city 1's stand-alone claim is 8 and city 2's is 4; the grand coalition's claim is $8+2=10$. In the two-player case, the Shapley value weights each coalition by $1/2$, so city 1's utility under the solution is $\frac{10+8-4}{2}=7$ and city 2's is $\frac{10+4-8}{2}=3$. Thus the outcome is $x^*=(7,3)$: the plant is built at location $A$ and city 1 transfers one monetary unit to city 2. 

We provide a normative foundation for this solution by showing that it is the only one to satisfy certain desirable properties: efficiency, anonymity, the dummy property, continuity, additivity, and monotonicity. Efficiency, anonymity, the dummy property, and continuity are standard \parencite{moulin:1988,thomson:2001}. Additivity captures the idea that the solution to two independent problems should not depend on whether they are considered jointly or separately \parencite{myerson:1981,perles:1981}. Monotonicity states that increasing a given player's utility in an allocation of the bargaining set cannot harm that player \parencite{kalai:1975aa}. Theorem 1 applies in Green's original model \parencite*{green:1983}. Theorem 2 applies in an extended model in which an alternative can be associated with a surplus or deficit in addition to a utility profile \parencite{green:2005}. In this case, we require a strengthening of the dummy property we call the axiom of compromise.

Our proof approach builds on the contributions of \textcite{green:1983,chambers:2005}. In particular, we use a result of \textcite{chambers:2005} that establishes that additive and Lipschitz-continuous solutions admit integral representations. The key insight is that monotonicity, when applied to the integral representation, entails the \emph{stand-alone approach}: the solution to a given problem is determined by the stand-alone claims corresponding to the problem. 

This is the first axiomatisation of the Shapley value of the stand-alone game, whether in the general framework of Green or in the specific classes of problems discussed above. Of course, axiomatisations of the Shapley value for coalitional games with transfers exist \parencite{shapley:1953,young:1985,hart:1989}, but these do not translate to the problem of social choice with transfers. To see why, recall that in coalitional games with transfers, the Shapley value is the unique solution satisfying efficiency, symmetry, the dummy property, and additivity. However in the social-choice model, many solutions satisfy these axioms, some of which are highly unnatural. The reason that uniqueness does not obtain is that social-choice problems are richer than coalitional games, in the sense that coalitional games are defined by a single number for each coalition, whereas social-choice problems are defined by an arbitrary bargaining set. 

The present paper is also the first to axiomatise a solution in the $n$-player problem of social choice with transfers. \textcite{green:1983,chambers:2005} focus on families of additive solutions, while \textcite{green:2005,chambers:2005a} restrict attention to two players. The branch of the literature following \textcite{moulin:1985a} studies a related but distinct problem, in which players have preferences over a fixed set of public actions, which play a fundamental role in the theory.\footnote{\cite{moulin:1985,moulin:1987,chun:1986,chun:1989,chun:2000}. We discuss the relation between the two branches of the literature in \cref{app_welfarism}.}

\section{Model, solution, and applications}\label{sec_model}

An \emph{$n$-player social-choice problem}, or \emph{bargaining set}, is a set $B\subset\mathbb{R}^{n}$ that is nonempty, closed, convex, comprehensive, and bounded above.\footnote{A subset $X$ of $\mathbb{R}^n$ is \emph{comprehensive} if, for every $x\in X$ and $y \in \mathbb{R}^{n}$, $y\le x$ implies $y\in X$; it is \emph{bounded above} if there exists $y\in\mathbb{R}^{n}$ such that $x\le y$ for every $x\in X$.} Let $\mathcal{B}$ be the set of such problems. An element of $ B$ is an \emph{allocation}, representing the utility of each player. Assuming the set is convex amounts to allowing randomisation; assuming it is comprehensive amounts to allowing free disposal. Let $N=\{1,2,\ldots,n\}$. Given $x\in\mathbb{R}^{n}$, let $x_S=\sum_{i\in S}x_i$ for each $S \subseteq N$, with the convention $x_\emptyset = 0$. Agents have quasi-linear preferences in the numéraire good and transfers are possible, so the \emph{feasible set} of allocations is $\{x\in\mathbb{R}^n:x_N\le\max_{y\in B}y_{N}\}$.

A \emph{solution} is a map $\psi: \mathcal{B} \rightarrow \mathbb{R}^n$ such that $\psi_N(B)\le \max_{x\in B}x_{N}$ for every $B\in\mathcal{B}$. The \emph{Shapley value of the stand-alone game} is the solution given by 
	\begin{equation}\label{eq:solution}
		\psi_i(B)=\frac 1 n\sum_{S\subseteq N\setminus i} {n-1 \choose |S|}^{-1}\left(\sup_{x \in B} x_{S\cup i}  -\sup_{x \in B} x_{S} \right)
	\end{equation}
for every $i\in N$ and $B \in \mathcal{B}$. 

One can think of $\sup_{x\in B}x_S$ as the \emph{stand-alone claim} of coalition $S$: it is the maximal total utility $S$ can obtain in $B$ while ignoring the other players. The \emph{stand-alone game} associated with a bargaining problem is the coalitional game with transfers in which the worth of a coalition equals its claim in the bargaining set. The solution is the Shapley value of this game. We will refer to $\sup_{x \in B} x_{S\cup i}  -\sup_{x \in B} x_{S}$ as the \emph{marginal stand-alone claim}, or simply the \emph{marginal claim}, of player $i$ for coalition $S$.

We now show that the model generalises several classes of problems and we apply the solution to a problem from each class. In some applications, alternatives are associated with surpluses or deficits that are not assigned to particular players. For example, the power plant may have to be funded by the cities themselves rather than central government, and different locations may entail different costs. In \cref{sec_surplus} we extend the model to cover such applications. 

Given a nonempty and bounded above $X\subset\mathbb{R}^{n}$, let $\ch(X)$ denote the smallest set in $\mathcal{B}$ that contains $X$; that is, $\ch(X)$ is the closed, convex, and comprehensive hull of $X$.

\begin{example}[Building shared infrastructure]
	Consider the construction of shared infrastructure, such as the power plant of the introduction. The regulatory frameworks of both the European Union and the United States mandate that the costs of shared energy infrastructure be split in proportion to the benefits \parencite{konstantelos:2017,hogan:2018}. However, these frameworks do not consider benefits that would obtain under alternatives that were not pursued. As a result, the final net benefit received by each party depends on the alternative chosen, which could lead to disagreements between parties. In contrast, the framework of the present paper offers a principled approach to compensating parties for forgone benefits. 

	Specifically, in the example of the power plant, the bargaining set is 
	\begin{equation*}
		B=\ch\left\{(8,2),(2,4)\right\}.
	\end{equation*}
	Then 1's marginal claims are 8 for $\emptyset$ and $10-4=6$ for $\{2\}$, so that $\psi_1(B)=7$. Similarly, 2's marginal claims are 4 for $\emptyset$ and $10-8=2$ for $\{1\}$, so that $\psi_2(B)=3$. Hence $\psi(B)=(7,3)$.
\end{example}

\begin{example}[Managing a common resource]
	Consider countries sharing a freshwater source. Half of the world's population lives within a transnational river basin, and the management and fair allocation of such resources are increasingly critical \parencite{turgul:2024}. The relevant United Nations convention hinges on the principle of `equitable and reasonable' use. However, the implementation of this principle is open to interpretation, which leads to disputes \parencite{mcintyre:2013}. Adopting a precise definition of the principle could avert such disputes. 
	
	To illustrate, suppose three countries rely on a river for freshwater. They each have equal rights to the use of the river, subject to the  constraint that the total amount drawn is at most $\bar w$ per year. If country $i$ draws $w_i$, it derives a benefit of  $\alpha_i\sqrt{w_i}$ in monetary terms, where $\alpha_i>0$. The bargaining set is
	\begin{equation*}
		B=\ch\left\{(\alpha_1\sqrt{w_1},\alpha_2\sqrt{w_2},\alpha_3\sqrt{w_3}): w\in\mathbb{R}^3_+ \text{ and } w_N\le \bar w\right\}.
	\end{equation*}
	The stand-alone claim of coalition $\{i\}$ is $\alpha_i \sqrt{\bar w}$. Solving the appropriate constrained optimisation problem, we find that the stand-alone claim of coalition $\{i,j\}$, where $j\ne i$, is $\sqrt{\bar w(\alpha_i^2+\alpha_j^2)}$. Similarly, the maximum total surplus in $B$ is $\sqrt{\bar w(\alpha_1^2+\alpha_2^2+\alpha_3^2)}$.

	Suppose $\bar w= 100$ and $\alpha_i=i$ for each country $i$. Then we obtain $\psi(B)\approx (4.5,11.7,21.3)$. The efficient water allocation is approximately $(7.1,28.6,64.3)$; to achieve the utilities, country 3 makes monetary transfers of approximately 1.8 to country 1 and 1.0 to country 2. 
\end{example}

\begin{example}[Settling debts]
	Consider the problem of how to compensate creditors following a bankruptcy. In modern law, creditors within the same class are paid in proportion to their claim. However, other solutions are also used, both in practice and in the theoretical literature; the one considered in the present paper is prominent among these \parencite{oneill:1982,aumann:1985a,thomson:2003}. 

	To illustrate how the solution applies to such problems, suppose a bankrupt company owes 100, 200, and 400 to companies 1, 2, and 3, respectively; however, its assets are only worth 300. We model this situation as the social-choice problem given by
	\begin{equation*}
		B=\ch\{x\in\mathbb{R}^3_+: x_1\le 100, \text{ } x_2\le 200, \text{ } x_3\le 400, \text{ and } x_1+x_2+x_3\le 300\}.
	\end{equation*}
	Then company 1's marginal claim for coalitions $\emptyset$ and $\{2\}$ is 100, and its marginal claim for other coalitions is zero. Hence $\psi_1(B)=50$. Company 2's marginal claim for coalitions $\emptyset$ and $\{1\}$ is 200, and its marginal claim for other coalitions is zero. Hence $\psi_2(B)=100$. Company 3 receives the residual; that is, $\psi_3(B)=150$.
\end{example}

\begin{example}[Dividing a good] 
	Consider the problem of allocating an indivisible good to which multiple parties have a claim. For instance, a bequest might leave a property in joint ownership. One of the joint owners may value owning the property outright more than its market value, and may wish buy out the others. In practice, the usual solution to this problem is to have the property valued and for one of the owners to buy the others' shares at the market value. However, this is problematic if multiple owners value the property at more than its market value \parencite{chang:2014}. Instead, the approach advanced in the present paper is for compensation to depend not on the market value of the property, but on the value each party attaches to owning the property outright. 

	For the sake of concreteness, suppose three siblings have inherited equal shares in a house. Siblings 1, 2, and 3 attach value 300, 270, and 240 respectively to owning the house outright.\footnote{This example is adapted from \textcite[chap.\ 9]{young:1994a}.} (Perhaps sibling 1 attaches a higher sentimental value to the house than sibling 2, while sibling 3 would sell the house.) It is socially efficient for sibling 1 to receive the house; the question is how siblings 2 and 3 should be compensated.
	
	The bargaining set is 
	\begin{equation*}
		B=\ch\{(300,0,0),(0,270,0),(0,0,240)\}.
	\end{equation*}
	The marginal claim of sibling 3 for coalition $S\subseteq N$ is 240 if $S=\emptyset$ and 0 otherwise. Hence 
	\begin{equation*}
		\psi_3(B)=\frac 1 3 {2 \choose 0}^{-1}240=80.
	\end{equation*}
	The marginal claim of sibling 2 for coalition $S\subseteq N$ is 270 if $S=\emptyset$, 30 if $S=\{3\}$, and 0 otherwise. Hence
	\begin{equation*}
		\psi_2(B)=\frac 1 3 {2 \choose 0}^{-1}270+\frac 1 3 {2 \choose 1}^{-1}30=95.
	\end{equation*}
	Finally, sibling 1 receives the residual; that is,
	\begin{equation*}
		\psi_1(B)=300-80-95=125.
	\end{equation*}
	Thus sibling 1 receives the house and transfers 95 to sibling 2 and 80 to sibling 3. 
\end{example}

\section{Axioms and result}
\label{sec_main}

We consider the following axioms. 

\begin{axiom}[Efficiency]\label{ax_efficiency}
	For every $B\in\mathcal{B}$ and $x\in B$, $\psi_N(B)\ge x_N$. 
\end{axiom}

Given a permutation $\rho : N \to N$, $x \in \mathbb{R}^n$, and $B \in \mathcal{B}$, let $\rho(x) = (x_{\rho(i)})_{i \in N}$ and $\rho(B) = \{\rho(y) : y \in B\}$. 

\begin{axiom}[Anonymity]\label{ax_anonymity}
	For every permutation $\rho : N \to N$ and $B \in \mathcal{B}$, $\psi(\rho(B)) = \rho(\psi(B))$.
\end{axiom}

Given $X\subseteq\mathbb{R}^n$, let $X^*$ denote its strong Pareto frontier; that is,
\begin{equation*}
	X^*=\{x\in X: \text{there is no $y\in X$ such that $y\ge x$ and $y\ne x$}\}.
\end{equation*}

\begin{axiom}[Dummy property]\label{ax_dummy}
	For every $B \in \mathcal{B}$ and $i\in N$ such that $x_i=0$ for every $x\in B^*$, $\psi_i(B)= 0$. 
\end{axiom}

\Cref{ax_efficiency,ax_anonymity,ax_dummy} are standard in the literature on cooperative game theory \parencite{moulin:1988,thomson:2001}.\footnote{In coalitional games, a dummy player is a player whose addition to any coalition does not change its value \parencite[119]{moulin:1988}. Here, a dummy player is a player who receives zero utility in every efficient allocation.}

\begin{axiom}[Continuity]\label{ax_continuity}
	$\psi$ is Lipschitz continuous with respect to the Hausdorff metric on $\mathcal{B}$.
\end{axiom}
\Cref{ax_continuity} is stronger than typical; we require Lipschitz continuity in order to appeal to a result of \textcite{chambers:2005}.

\begin{axiom}[Additivity]\label{ax_additivity}
	For every $B,C \in \mathcal{B}$, $\psi(B+C) = \psi(B) + \psi(C)$.
\end{axiom}

Set addition is meant in the sense of Minkowski. \Cref{ax_additivity} implies that solving two independent problems jointly should yield the same solution as if they were considered separately. By independent problems, we mean that the choice of alternative in one problem does not affect the utilities of the alternatives in the other problem. For example, suppose the locations of both a power plant and an airport are to be determined. It may be that the benefits that accrue from the different choices of location for the airport are not affected by the choice of location for the plant, and vice versa. If two problems are independent in this sense, then the set of utility profiles that can be achieved by considering the two problems jointly (without transfers) is the sum of the individual bargaining sets. In this case, additivity imposes that the solution to the joint problem be the sum of the solutions to the individual problems. However, if two problems are not independent, then the bargaining set of the joint problem does not equal the sum of the individual bargaining sets. For example, it may be that there is a regulatory minimum distance between the airport and the power plant and that some of the possible location combinations breach this distance. In this case, additivity does not bite.\footnote{For further discussion see \cite{roth:1979,myerson:1981,perles:1981,peters:1986}.}

By a slight abuse of notation, given $X \subseteq \mathbb{R}^n$ and $x\in\mathbb{R}^n$, we will write `$X\cup x$' to mean $X\cup\{x\}$.\footnote{Analogously, we will let `$X + x$' mean $\{y+x: y \in X\}$. If $S\subseteq N$ and $i\in N$ we will let `$S\cup i$' mean $S\cup\{i\}$ and `$S\setminus i$' mean $S\setminus\{i\}$.}  If $X\subset\mathbb{R}^n$ is nonempty and bounded above, then `$\psi(X)$' is understood to mean $\psi(\ch(X))$.

\begin{axiom}[Monotonicity]\label{ax_mon}
	For every $B \in \mathcal{B}$, $i \in N$, and $x\in\mathbb{R}^n$ such that   there exists $y\in B$ such that $x_i \ge y_i$ and $y_j = x_j$ for each $j \ne i$, $\psi_i(B\cup x) \ge \psi_i(B)$.
\end{axiom}

\Cref{ax_mon} states that improving a player's payoff in an allocation of the bargaining set, while keeping all other players' payoffs the same, should weakly benefit that player. A number of monotonicity properties have been studied in the literature on cooperative game theory; ours is most similar to that of \textcite{kalai:1975aa}.\footnote{See also \cite{roth:1979a,thomson:1980,young:1985,moulin:1988,chun:1989}.}

Our main result is the following.

\begin{theorem}
	\label{theorem:solution}
	A solution on $\mathcal{B}$ satisfies efficiency, anonymity, the dummy property, continuity, additivity, and monotonicity if and only if it is the Shapley value of the stand-alone game. 
\end{theorem}
In \cref{sec_existence} we prove that the Shapley value of the stand-alone game satisfies \cref{ax_efficiency,ax_additivity,ax_continuity,ax_dummy,ax_anonymity,ax_mon}; in \cref{sec_uniqueness} we prove that it is the only solution to do so. It is especially interesting that this solution emerges despite the fact that coalitions play no explicit role either in the model (unlike in coalitional games, subcoalitions cannot achieve anything severally from the grand coalition) or in any of the axioms.

\subsection{Proof that the solution satisfies the axioms}\label{sec_existence}

Efficiency, anonymity, and the dummy property follow from standard arguments that establish these properties for the Shapley value. Continuity and additivity follow from the fact that, for each $S \subseteq N$, the map $B \mapsto \sup_{x\in B}x_S$ is linear and Lipschitz-continuous on $\mathcal{B}$. 

To establish monotonicity, fix $B\in\mathcal{B}$, $i\in N$, and $x\in\mathbb{R}^n$ satisfying the hypotheses of the axiom. Then, for each $S\subseteq N\setminus i$,
\begin{equation*}
	\begin{aligned}
		&\sup_{z\in B\cup x}z_{S\cup i}\ge \sup_{z\in B}z_{S\cup i} \text{ and}\\
		&\sup_{z\in B\cup x}z_{S} = \sup_{z\in B}z_{S}.
	\end{aligned}
\end{equation*}
Hence $\psi_i(B\cup x)\ge \psi_i(B)$.

\subsection{Proof that no other solution satisfies the axioms}
\label{sec_uniqueness}
Write $\mathbf{0}^n$ and $\mathbf{1}^n$ for the vectors in $\mathbb{R}^n$ consisting entirely of zeroes and of ones, respectively. Let 
\begin{equation*}
	\begin{aligned}
		\mathcal{B}_0&=\Big\{B \in \mathcal{B} : \max_{x \in B}x_N = 0\Big\},\\
		\mathbb{S}_+ &= \{x \in \mathbb{R}^n_+ : \|x\| = 1\} \setminus \left\{\sfrac{\mathbf{1}^n}{\sqrt{n}}\right\}, \text{and}
		\\ \mathbb{S}_0 &= \left\{x \in \mathbb{R}^n : \|x\| = 1 \text{ and }x_N = 0\right\},
	\end{aligned}
\end{equation*}
where $\|\cdot\|$ denotes the Euclidean norm on $\mathbb{R}^n$. \Cref{fig_n2} shows $\mathbb{S}_+$ and $\mathbb{S}_0$ for $n=2$.

Let $\mathcal{N}$ be the set of nonempty proper subsets of $N$, and, for each $S \in \mathcal{N}$, let $v^S\in \mathbb{S}_+$ be given by
\begin{equation*}
	v^S_i=
	\begin{cases}
		\frac{1}{\sqrt{|S|}} & \text{if }i\in S \text{ and}\\
		0 & \text{otherwise.}
	\end{cases}
\end{equation*}
By a slight abuse of notation, we will write `$v^i$' to mean $v^{\{i\}}$.

\begin{figure}[t]
	\definecolor{my_red}{RGB}{195,59,45}
	\definecolor{my_yellow}{RGB}{241,196,79}
	\definecolor{my_blue}{RGB}{45,97,138}
	\centering
	\begin{tikzpicture}[scale=1.4]
		\draw[very thick,color=my_blue] (4-2*2^.5,-4) -- (4-2*2^.5,-2*2^.5) arc[start angle=0, end angle=90, radius=4] -- (-4,4-2*2^.5);
		\fill[color=my_blue!40] (4-2*2^.5,-4) -- (4-2*2^.5,-2*2^.5) arc[start angle=0, end angle=90, radius=4] -- (-4,4-2*2^.5) -- (-4,-4) -- cycle;
		\draw [->] (-4,0) -- (5,0) node [right] {$x_1$};
		\node [above right] at (0,0) {$0$};
		\draw [->] (0,-4) -- (0,5) node [right] {$x_2$};
		\draw [dashed] (4,-4) -- (-4,4) ;
		\draw[very thick,color=my_red] (4,0) arc[start angle=0, end angle=90, radius=4];
		\draw[color=my_red,fill=white, very thick] (45:4) circle [radius=.05] node[above right,black] {$\sfrac{\mathbf{1}^n}{\sqrt{n}}$};
		\draw[color=my_yellow,fill] (135:4) circle [radius=.05] node[above right] {$\mathbb{S}_0$};
		\draw[color=my_yellow,fill] (-45:4) circle [radius=.05];
		\node[my_blue] at (-2,-2) {$B_0$};
		\draw[fill] (25:4) circle [radius=.05] node[above right] {$v$};
		\node [above right,color=my_red] at (70:4) {$\mathbb{S}_+$};
		\draw[fill] ([shift=(25:4.15)] -2*2^.5,-2*2^.5)  circle [radius=.05] node[above right] {$\xi^\alpha(v)$};
	\end{tikzpicture}
    \caption[]{$\mathbb{S}_+$ and $\mathbb{S}_0$ for $n=2$, and $B_0$ and $\xi^\alpha(v)$ for $S = N$.}
    \label{fig_n2}
\end{figure}

The following result is a straightforward variant of Chambers and Green's  \parencite*{chambers:2005} theorem 1.

\begin{lemma}\label{lemma:representation}
	A solution $\psi$ satisfies efficiency, continuity, and additivity if and only if there exist a finite Borel measure $\mu$ on $\mathbb{S}_+$ and a Borel-measurable function $h : \mathbb{S}_+ \to \mathbb{S}_0$ such that
	\begin{equation}
		\label{eq:representation}
		\psi(B) = \int \left(\sup_{x \in B} x \cdot v\right) h(v) \diff \mu(v) 
	\end{equation}
	for every $B \in \mathcal{B}_0$.
\end{lemma}
One can interpret each $v\in \mathbb{S}_+$ as defining a weighted coalition, where $v_i$ is the weight placed on player $i$, and $\sup_{x\in B}x\cdot v$ as being the claim of the weighted coalition; then $v^S$ corresponds to the unweighted coalition $S$ and $\sup_{x\in B}x\cdot v^S$ is the claim of coalition $S$. \Cref{lemma:representation} does not impose any restrictions on $\mu$ other than finiteness nor on $h$ other than measurability. There is therefore a broad class of solutions satisfying efficiency, continuity, and additivity. As discussed by \textcite{chambers:2005}, anonymity and the dummy property impose some restrictions, but the class of solutions satisfying \cref{ax_efficiency,ax_additivity,ax_continuity,ax_dummy,ax_anonymity} is still broad, and includes some that are highly unnatural. We will see that adding monotonicity rules out such solutions.

In what follows, fix a solution $\psi$ satisfying \cref{ax_efficiency,ax_additivity,ax_continuity,ax_dummy,ax_anonymity,ax_mon} and $\mu$ and $h$ satisfying the hypotheses of \cref{lemma:representation}. It suffices to prove that $\psi$ is uniquely determined on $\mathcal{B}_0$ by additivity and the fact that, for every $x\in\mathbb{R}^n$, 
\begin{align}
	\psi(\{x\})&=\sum_{i\in N}x_i\psi(\{v^i\})\nonumber\\
	&=\sum_{i\in N}x_iv^i\nonumber\\
	&=x,\label{eqn:trans_inv}
\end{align}
where the first equality follows from additivity and the second from efficiency and the dummy property.

Let $\mathbb{B}_+(v,r) = \{w \in \mathbb{S}_+ : \|w-v\| \le r\}$ for every $v\in\mathbb{S}_+$ and $r>0$. We will work with a set of $v\in \mathrm{supp}(\mu)$ such that $h(v)$ can be approximated by the $\mu$-weighted average of $h$ in a neighbourhood of $v$.
Let $\mathbb{L}$ be the set of $\mu$-Lebesgue points of $h$; that is, the set of $v \in \mathrm{supp}(\mu)$ satisfying 
\begin{equation*}
	\lim_{r \downarrow 0} \frac{\int_{\mathbb{B}_+(v,r)} \|h(w)-h(v)\| \diff \mu(w)}{\mu(\mathbb{B}_+(v,r))} = 0.
\end{equation*}
The Lebesgue differentiation theorem ensures that $\mu$-almost every $v \in \mathbb{S}_+$ lies in $\mathbb{L}$ \parencite[corollary 2.9.9]{federer:1969}.

The next lemma establishes that $\mu$-Lebesgue points of $h$ must correspond to unweighted coalitions. 

\begin{lemma}
	\label{lemma:vs}
	$\mathbb{L} \subseteq \{v^S : S \in \mathcal{N}\}$.
\end{lemma}
\begin{proof}
	Fix $v \in \mathbb{L}$ and, seeking a contradiction, suppose that there is no $S \in \mathcal{N}$ such that $v = v^S$. Let $S = \{i \in N : v_i > 0\}$. We may choose $v$ such that $S$ is maximal; that is, such that $S \subset \{i \in N : w_i > 0\}$ for some $w \in \mathbb{L}$ only if $w = v^T$ for some $T \in \mathcal{N}$.

	Let 
	\begin{equation*}
		\begin{aligned}
			H &= \{x \in \mathbb{R}^n : x_i = 0 \text{ for each } i \in N \setminus S\} \text{ and}
			\\B_0 &= \ch\left(\mathbb{S}_+ \cap H\right)-\frac{\sqrt{|S|}}{n}\mathbf{1}^n,
		\end{aligned}
	\end{equation*}
	noting that $B_0$ is well-defined, since $\mathbb{S}_+ \cap H$ is nonempty and bounded above, and lies in $\mathcal{B}_0$.
	Define
	\begin{equation*}
	\begin{aligned}
		\xi^\alpha(w) &= (1+\alpha)w - \frac{\sqrt{|S|}}{n}\mathbf{1}^n \text{ and} \\
		B_{\alpha r} &= \ch\left(B_0 \cup \xi^\alpha\left(\mathbb{B}_+(v,r) \cap H\right)\right)
	\end{aligned}
	\end{equation*}
	for every $\alpha,r > 0$ and $w \in \mathbb{S}_+$, noting that $B_{\alpha r} \in \mathcal{B}_0$ for sufficiently small $\alpha$ and $r$, since $v \ne v^S$.

	\emph{Step 1:} We first show that the solutions at $B_0$ and $B_{\alpha r}$ are equal for sufficiently small $\alpha$ and $r$. The argument uses monotonicity to show that each $i\in S$ must be weakly better off under $B_{\alpha r}$ compared to $B_0$, and the dummy property to show that each $i\notin S$ must be unaffected. 

	For each $i \in N \setminus S$, 
	\begin{equation*}
	\begin{aligned}
		\psi_i(B_{\alpha r}) + \frac{\sqrt{|S|}}{n} &= \psi_i\left(B_{\alpha r} + \frac{\sqrt{|S|}}{n}\mathbf{1}^{n}\right) = \psi_i\left(H \cap \left(\mathbb{S}_+ \cup \left((1+\alpha)\mathbb{B}_+(v,r)\right)\right)\right)
		\\&= 0 = \psi_i\left(\mathbb{S}_+ \cap H\right) = \psi_i\left(B_{0} + \frac{\sqrt{|S|}}{n}\mathbf{1}^{n}\right) = \psi_i\left(B_0\right) + \frac{\sqrt{|S|}}{n}
	\end{aligned}
	\end{equation*}
	where the first and last equalities follow from efficiency, anonymity and additivity, and the third and fourth follow from the dummy property.
	By efficiency, it is therefore enough to show that, given $i \in S$ and sufficiently small $\alpha$ and $r$, $\psi_i(B_{\alpha r}) \ge \psi_i(B_0)$.
	In turn, by continuity and monotonicity, it suffices to exhibit, given $i \in S$ and sufficiently small $\alpha$ and $r$, and for all $x \in B_{\alpha r}$, a $y \in B_0$ such that $y_i \le x_i$ and $y_j = x_j$ for all $j \in N \setminus i$.\footnote{To see why this suffices, choose $\{x^k\}_{k = 1}^\infty \subset B_{\alpha r}$ such that $\ch(B_0 \cup \{x^k\}_{k = 1}^\infty) = B_{\alpha r}$.
	For all $k \in \mathbb{N}$, choose $y^k \in B_0$ such that $y^k_i \le x^k_i$ and $y^k_j = x^k_j$ for each $j \in N \setminus i$, and note that $\psi_i(B_0) \le \psi_i(B_0 \cup x^1) \le \psi_i(B_0 \cup \{x^1,x^2\}) \le \dots$, by monotonicity.
	Then, $\psi_i(B_{\alpha r}) = \lim_{m \to \infty} \psi_i(B_0 \cup \{x^k\}_{k = 1}^m) \ge \psi_i(B_0)$, where the first equality follows from continuity. 
	}
	To this end, fix $i \in S$ and note first that, given sufficiently small $\alpha$ and $r$ and for every $w \in \mathbb{B}_+(v,r) \cap H$, there exists $\hat w \in \mathbb{S}_+ \cap H$ such that $\hat w_i \le (1+\alpha)w_i$ and $\hat w_j = (1+\alpha)w_j$ for each $j \in N \setminus i$.
	Now fix such $\alpha$ and $r$, as well as $x \in B_{\alpha r}$.
	Choose $w \in \mathbb{B}_+(v,r) \cap H$, $u \in B_0$, and $\beta \in [0,1]$ such that $x \le \beta \xi^\alpha(w) + (1-\beta)u$, and let $y=x+\beta(\hat w-(1+\alpha)w)$. Note that 
	\begin{equation*}
		y \le y -x + \beta \xi^\alpha(w) + (1-\beta)u = \beta\left(\hat w -\frac{\sqrt{|S|}}{n} \mathbf{1}^n\right) + (1-\beta)u \in B_0,
	\end{equation*}
	so that $y \in B_0$, and that $y_i \le x_i$, as desired.

	\emph{Step 2:} Next, we show that the suitably normalised difference in the solutions at $B_0$ and $B_{\alpha r}$ equals $h(v)$ in the limit as $r$ and $\alpha$ vanish. The argument relies on the fact that $v$ is a $\mu$-Lebesgue point of $h$. 

	Let $\Delta_{\alpha r} : \mathbb{S}_+ \to \mathbb{R}$ be given by 
	\begin{equation*}
		\Delta_{\alpha r}(w) = \frac 1 \alpha \left(\sup_{x \in B_{\alpha r}} x \cdot w - \sup_{x \in B_0} x \cdot w\right)
	\end{equation*}
	for all $\alpha, r > 0$, and $\pi : \mathbb{S}_+ \to \mathbb{R}^n$ be given by 
	\begin{equation*}
		\pi_i(w) = 
		\begin{dcases}
			\frac{w_i}{\sqrt{\sum_{j \in S}w_j^2}} & \text{if $i \in S$ and $w_j > 0$ for some $j \in S$ and} \\
			0 & \text{otherwise.}
		\end{dcases}
	\end{equation*}
	Note that, for every $w \in \mathbb{S}_+$,
	\begin{equation*}
		\sup_{x \in B_0} x \cdot w = \left(\pi(w)-\frac{\sqrt{|S|}}{n}\mathbf{1}^n\right)\cdot w= \sqrt{\sum_{i \in S} w_i^2}-\frac{\sqrt{|S|}}{n}w_N,
	\end{equation*} 
	so that  
	\begin{equation}
		\label{eq:dB}
		\Delta_{\alpha r}(w) = \frac{1}{\alpha}\sup\left\{\sqrt{\sum_{i \in S} w_i^2}\left((1+\alpha)\sup_{u \in \mathbb{B}_+(v,r)} u \cdot \pi(w) - 1\right),0\right\}.
	\end{equation}
	Let
	\begin{equation*}
		W = \{w \in \mathbb{S}_+ : w_i = 0 \text{ for some } i \in S\} \cup \{v^T : S \subset T \subset N\}
	\end{equation*} 
	and note that $\mathbb{B}_+(v,r)$ and $\pi(W)$ are closed and disjoint for $r$ sufficiently small (since $v \ne v^S$), so that $u \cdot \pi(w)$ is bounded away from $1$ across all $u \in \mathbb{B}_+(v,r)$ and $w \in W$.
	Thus, given sufficiently small $\alpha$ and $r$, $\Delta_{\alpha r}$ vanishes on $W$ and, hence, on $\mathbb{L} \setminus H$, since $S$ is maximal.

	Therefore,
	\begin{equation*}
	\begin{aligned}
		\lim_{r \downarrow 0} \lim_{\alpha \downarrow 0}\frac{\psi(B_{\alpha r})-\psi(B_0)}{\alpha\mu(\mathbb{B}_+(v,r))} 
		&= \lim_{r \downarrow 0} \frac 1 {\mu(\mathbb{B}_+(v,r))} \lim_{\alpha \downarrow 0} \int \Delta_{\alpha r}(w)  h(w) \diff\mu(w)
		\\&= \lim_{r \downarrow 0} \frac 1 {\mu(\mathbb{B}_+(v,r))} \lim_{\alpha \downarrow 0} \int_H \Delta_{\alpha r}(w)  h(w) \diff\mu(w)
		\\&= \lim_{r \downarrow 0} \frac 1 {\mu(\mathbb{B}_+(v,r))} \int_{\mathbb{B}_+(v,r) \cap H} h(w) \diff\mu(w)
		\\&= \lim_{r \downarrow 0} \frac 1 {\mu(\mathbb{B}_+(v,r))} \int_{\mathbb{B}_+(v,r)} h(w) \diff\mu(w)
		\\& = h(v),
	\end{aligned}
	\end{equation*}
	where the first equality follows from \cref{lemma:representation}, since $B_0 \in \mathcal{B}_0$ and $B_{\alpha r} \in \mathcal{B}_0$ for $\alpha$ and $r$ sufficiently small;
	the second holds since $\Delta_{\alpha r}$ vanishes on $\mathbb{L} \setminus H$ for sufficiently small $\alpha$ and $r$, and $\mu$ is concentrated on $\mathbb{L}$; 
	the third holds by the bounded convergence theorem, since $\Delta_{\alpha r}$ is bounded above by $\alpha$, equal to $\alpha$ on $\mathbb{B}_+(v,r) \cap H$, and, for every $w \in (\mathbb{S}_+ \cap H) \setminus \mathbb{B}_+(v,r)$, equal to zero for $\alpha$ sufficiently small, by \eqref{eq:dB};
	the fourth holds since $\mathbb{B}_+(v,r) \cap \mathbb{L} \subset H$ for $r$ sufficiently small, since $S$ is maximal, and $\mu$ is concentrated on $\mathbb{L}$;
	and the fifth holds since $v \in \mathbb{L}$.

	Steps 1 and 2 together imply that $h(v)=\mathbf{0}^n$, contradicting the fact that $h(v) \in \mathbb{S}_0$.
\end{proof}

\Cref{lemma:representation,lemma:vs} together imply that there exist weights $(w_S)_{S \in \mathcal{N}} \subset \mathbb{R}_+$ and vectors $(h^S)_{S \in \mathcal{N}} \subset \mathbb{S}_0$ such that  
\begin{equation}
	\label{eq:representation_finite}
	\psi(B)=\sum_{S\in \mathcal{N}} \left(\sup_{x \in B} x_{S} \right) w_S h^S
\end{equation}
for every $B\in \mathcal{B}_0$. To complete the proof, it remains to determine these weights and vectors.

Fix $S \in \mathcal{N}$ and define $B'_\alpha = \ch(\mathbb{S}_+ \cup \{(1+\alpha)v^S\}) - \frac{1}{\sqrt{n}}\mathbf{1}^n$ for all $\alpha > 0$. 
For sufficiently small $\alpha$, $B'_\alpha \in \mathcal{B}_0$ and, by \eqref{eq:representation_finite}, 
\begin{equation}
	\psi(B'_\alpha) - \psi\left(B'_0\right) = \alpha \sqrt{|S|}w_S h^S\label{eqn:S_set}.
\end{equation}
Therefore, by anonymity, $h^S_i = h^S_j$ whenever either $i,j\in S$ or $i,j\notin S$, so that
\begin{equation}
	\label{eq:hS}
	h^S_i=
	\begin{cases}
		\sqrt{\frac{n-|S|}{n|S|}} & \text{if }i\in S \text{ and}\\
		-\sqrt{\frac{|S|}{n(n-|S|)}} & \text{otherwise.}
	\end{cases}
\end{equation}
Hence, by \eqref{eqn:S_set} and anonymity, there exist $(\bar w_k)_{k = 1}^{n-1}$ such that $w_S =\bar w_{|S|}$ for each $S \in \mathcal{N}$.

If $n = 2$, taking $y = (1,-1)$ in \eqref{eqn:trans_inv} yields
\begin{equation*}
	1 = \psi_1(\{(1,-1)\}) = \bar w_1 h^{\{1\}}_1 - \bar w_1 h^{\{2\}}_1 = \sqrt{2}\bar w_1,
\end{equation*} 
so that $\psi$ is uniquely determined on $\mathcal{B}_0$.

Assume instead that $n > 2$.
Fix $i \in N$ and a nonempty $S \subset N\setminus i$, and let  
\begin{equation*}
	B''_\alpha = \ch\left(\left\{v^T : T \in \mathcal{N},\; i\notin T\right\} \cup \left\{(1+\alpha)v^S\right\}\right) - \frac{\sqrt{n-1}}{n}\mathbf{1}^n
\end{equation*}
for every $\alpha \ge 0$.
Choose $\bar \alpha > 0$ such that $v^{N \setminus i}_N > (1+\bar\alpha)v_N^S$, so that $B''_\alpha \in \mathcal{B}_0$ for every $\alpha \in[0, \bar\alpha]$.
Note also that $x_i = -\sqrt{n-1}/n$ for every $x \in B''^{*}_\alpha$.
Then
\begin{equation*}
	\psi_i(B''_\alpha) = \psi_i\left(B''_\alpha+\frac {\sqrt{n-1}}{n}v^i\right) - \psi_i\left(\left\{\frac {\sqrt{n-1}}{n}v^i\right\}\right) = -\frac{\sqrt{n-1}}{n}
\end{equation*}
for every $\alpha \in[0, \bar\alpha]$, where the first equality follows from additivity and the second from the dummy property and \eqref{eqn:trans_inv}. 
Fix $\alpha\in (0 , \bar \alpha]$ small enough that
\begin{equation*}
	\argmax_{x \in B''_\alpha} x_T = \argmax_{x \in B''_\alpha} x_{T \cup i} = 
	\begin{cases}
		\left\{(1+\alpha)v^S-\frac{\sqrt{n-1}}{n}\mathbf{1}^n\right\} & \text{if $T = S$ and}\\
		\left\{v^T-\frac{\sqrt{n-1}}{n}\mathbf{1}^n\right\} & \text{otherwise}
	\end{cases}
\end{equation*} 
for each nonempty $T \subseteq N \setminus i$. Then we have  
\begin{align}
	\nonumber 0 &= \psi_i(B''_\alpha) - \psi_i(B''_0)
	\\\nonumber &= \alpha v^S_{S \cup i}\bar w_{|S|+1} h_i^{S \cup i} + \alpha v^S_{S}\bar w_{|S|} h_i^S
	\\\label{eq:recursive}&= \alpha\sqrt{|S|} \left(\bar w_{|S|+1}\sqrt{\tfrac{n-(|S|+1)}{n(|S|+1)}}
	-\bar w_{|S|} 
	\sqrt{\tfrac{|S|}{n(n-|S|)}}
	\right).
\end{align}
Hence
\begin{equation}
	\label{eq:int}
	\psi_i(B)=
	\sum_{S\subset N \setminus i}\left(\sup_{x \in B} x_{S\cup i} -\sup_{x \in B} x_{S} \right) \bar w_{|S|+1}\sqrt{\tfrac{n-(|S|+1)}{n(|S|+1)}}
\end{equation}
for each $i\in N$ and every $B\in\mathcal{B}_0$ (where, recall, $x_\emptyset = 0$). 

Next, taking $y\in\mathbb{R}^n$ such that $y_i =1$ and $y_N=0$ yields
\begin{equation}\label{eq:total}
	1=\psi_i(\{y\})=\sum_{S\subset N\setminus i}\bar w_{|S|+1}\sqrt{\tfrac{n-(|S|+1)}{n(|S|+1)}},
\end{equation}
where the first equality follows from \eqref{eqn:trans_inv} and the second from \eqref{eq:int}. 

Equation \eqref{eq:total} and the $n-2$ equations obtained by varying $S$ in \eqref{eq:recursive} yield a unique list of weights $(\bar w_k)_{k = 1}^{n-1}$. Hence $\psi$ is uniquely determined on $\mathcal{B}_0$, as desired.

\section{Surpluses and deficits}
\label{sec_surplus}
In some applications, alternatives are associated with surpluses or deficits. For instance, the power plant may have to be funded by the cities themselves rather than from earmarked funds, and different locations may entail different costs. We extend the model to allow for this and discuss two applications.

Let $\mathcal{C}$ be the set of all nonempty, closed, convex, comprehensive, and bounded above subsets of $\mathbb{R}^{n+1}$, endowed with the Hausdorff metric. Following \textcite{green:2005}, we interpret each element of $\mathcal{C}$ as a bargaining set, with the first $n$ entries representing the players' utilities and the last entry representing the surplus at the alternative; a negative last entry represents a deficit. (Since $x\in\mathbb{R}^{n+1}$ is not strictly speaking an allocation, we use the term `alternative'.)

Let $N+1=\{1,2,\ldots,n+1\}$. A solution on $\mathcal{C}$ is a map $\phi:\mathcal{C}\to\mathbb{R}^n$ such that $\phi_N(C)\le \max_{x\in C}x_{N+1}$ for every $C\in\mathcal{C}$. Thus a surplus expands the feasible set, whereas a deficit shrinks it. Given a nonempty and bounded above $X\subset\mathbb{R}^{n+1}$, let $\ch(X)$ denote the smallest set in $\mathcal{C}$ that contains $X$. If $X\subset\mathbb{R}^{n+1}$ is nonempty and bounded above, then $\phi(X)$ is understood to mean $\phi(\ch(X))$. Write $\mathbf{0}^{n+1}$ and $\mathbf{1}^{n+1}$ for the vectors in $\mathbb{R}^{n+1}$ consisting entirely of zeroes, and of ones, respectively.

We consider analogues of the axioms of \cref{sec_main}, except that we impose a strengthening of the dummy property.

\begin{axprime}[Efficiency]\label{axprime_efficiency}
	For every $C\in\mathcal{C}$ and $x\in C$, $\phi_N(C)\ge x_{N+1}$. 
\end{axprime}

Given a permutation $\rho : N \to N$ and $C \in \mathcal{C}$, let
\begin{equation*}
	\rho(C) = \{(x_{\rho(1)},\dots,x_{\rho(n)},x_{n+1}) : x \in \mathcal{C}\}.
\end{equation*}

\begin{axprime}[Anonymity]\label{axprime_anonymity}
	For every $C \in \mathcal{C}$ and permutation $\rho : N \to N$, $\phi(\rho(C)) = \rho(\phi(C))$.
\end{axprime}

\begin{axprime}[Compromise]\label{axprime_compromise}
	For every $C \in \mathcal{C}$ and $i\in N$, 
	\begin{equation*}
		\min_{x\in C^*}x_i+\frac{x_{n+1}}{n}\le \phi_i(C)\le \sup_{x\in C}x_i+\frac{x_{n+1}}{n}.
	\end{equation*}
\end{axprime}

\Cref{axprime_compromise} is a strengthening of the dummy property. It states that a player should get at least as much as in their least preferred efficient alternative and at most as much as in their most preferred alternative, assuming any surplus or deficit is shared equally. The axiom of compromise is a combination of Moulin's \parencite*{moulin:1985} `minimal individual rationality' and `no free lunch' adapted to our model.
	
\begin{axprime}[Continuity]\label{axprime_continuity}
	$\phi$ is Lipschitz continuous with respect to the Hausdorff metric on $\mathcal{C}$.
\end{axprime}

\begin{axprime}[Additivity]\label{axprime_additivity}
	For every $B,C \in \mathcal{C}$, $\phi(B+C) = \phi(B) + \phi(C)$.
\end{axprime}

\begin{axprime}[Monotonicity]\label{axprime_mon}
	For every $C \in \mathcal{C}$, $i \in N$, and $x\in \mathbb{R}^{n+1}$ such that there exists $y \in C$ such that $x_i \ge y_i$ and $y_j = x_j$ for each $j \in N+1 \setminus \{i\}$, $\psi_i(C\cup x) \ge \psi_i(C)$.
\end{axprime}

\begin{theorem}
	\label{theorem:surplus}
	A solution $\phi$ on $\mathcal{C}$ satisfies efficiency, anonymity, compromise, continuity, additivity, and monotonicity if and only if
	\begin{equation}\label{eq:surplus}
		\phi_i(C)=\frac 1 n\sum_{S\subseteq N\setminus i} {n-1 \choose |S|}^{-1}\left(\sup_{x \in C} (x_{S\cup i}+x_{n+1} ) -\sup_{x \in C} (x_{S}+x_{n+1}) \right)
	\end{equation}
	for each $i\in N$ and every $C \in \mathcal{C}$.
\end{theorem}

The solution is analogous to the solution in \cref{theorem:solution}, except that the stand-alone claim of each coalition takes into account the surplus associated with different alternatives. In particular, the claim of a coalition ignores the other players: the surplus or deficit is fully captured by the coalition. 

We illustrate with two examples.

\begin{table}[t]
	\centering
		\begin{tabular}{llcccc}
		Size &Location &  Utility 1 &  Utility 2 & Cost \\ \hline
		\multirow{2}{*}{$L$} & $A$ &10 &4 & \multirow{2}{*}{4}\\
		&$B$&4&6&\\
		\multirow{2}{*}{$S$}& $A$ & 9 & 2 &\multirow{2}{*}{2}\\
		&$B$&2&7&
	\end{tabular}
	\caption{Building a power plant with costs and different sizes.}
	\label{tab_plants}
\end{table}

\begin{example}[Building a power plant]
	Suppose that the cost of the plant must be borne by the cities rather than from earmarked funds and that the government, in addition to deciding on the location, must decide whether to build a small $(S)$ or large $(L)$ plant. The utilities and costs of the different options are given in \cref{tab_plants}. The bargaining set is 
	\begin{equation*}
		C=\ch\{(10,4,-4),(4,6,-4),(9,2,-2),(2,7,-2)\}.
	\end{equation*}
	The claims are
	\begin{equation*}
	\begin{aligned}
		\sup_{x\in C}x_1+x_3&=9-2=7,\\	
		\sup_{x\in C}x_2+x_3&=7-2=5, \text{ and}\\
		\sup_{x\in C}x_1+x_2+x_3&=10+4-4=10.
	\end{aligned}
	\end{equation*}
	Hence the solution is $\phi(C)=(6,4)$. The power plant is built at location $A$ with a large capacity and city 1 pays the entire cost. 
\end{example}

\begin{example}[Providing a public good] 
	The mayor of a town has decided to build a new park, to be funded by a special levy. Out of concern for fairness, the mayor wants those who will benefit more from the park to pay more. The civil servants divide the \num{10000} households of the town into three groups, depending on their distance to the park, and estimate that the benefit of household $i$ is 
	\begin{equation*}
		v_i=
		\begin{cases}
			20 &\text{if } i\le \num{2000},\\
			10 &\text{if }\num{2000}<i\le \num{6000}, \text{ and}\\
			0 &\text{otherwise.}
		\end{cases}
	\end{equation*}
	The cost of building the park is $c=\num{40000}$. The bargaining set is 
	\begin{equation*}
		C=\ch\{(v_1,v_2,\ldots,v_n,-c),\mathbf{0}^{n+1}\}.
	\end{equation*}
	The marginal claim of player $i> \num{6000}$ is zero for any coalition $S\subseteq N\setminus i$, so $\phi_i(C)=0$ if $i>\num{6000}$. In general, the marginal claim of player $i$ for $S\subseteq N\setminus i$ is $\max\{v_S+v_i- \num{60000},0\}$. Since $n$ is large relative to $v_i$, the proportion of coalitions in which $i$ is pivotal is small, and so we can approximate the solution by ignoring the coalitions in which players are pivotal. For coalitions $S$ in which $i$ is not pivotal, her marginal claim is $v_i$ if $v_S\ge \num{60000}$ and zero otherwise. This implies that the solution awards approximately twice as much to households with value 20 than households with value 10. Since the net value created by the park is \num{40000}, we have 
	\begin{equation*}
		\phi_i(C)\approx
		\begin{cases}
			10 &\text{if } i\le \num{2000},\\
			5 &\text{if }\num{2000}<i\le \num{6000}, \text{ and}\\
			0 &\text{otherwise.}
		\end{cases}
	\end{equation*}
	Thus the solution yields the following intuitive outcome: the households with zero value for the park pay nothing, the households with value 20 pay approximately 10 each, and the households with value 10 pay approximately 5 each.
\end{example}

\subsection{Proof that the solution satisfies the axioms}
Efficiency, anonymity, continuity, additivity, and monotonicity follow straightforwardly from analogous arguments to those made in \cref{sec_existence}.

For compromise, fix $C\in\mathcal{C}$ and $i\in N$. Let $y$ and $z$ be a solutions to $\min_{x\in C^*}x_i+\frac{x_{n+1}}{n}$ and $ \sup_{x\in C}x_i+\frac{x_{n+1}}{n}$, respectively. The term in $\phi_i(C)$ corresponding to each $S\subseteq N\setminus i$ is at least $\min_{x\in C^*}x_i$. The term corresponding to $S=\emptyset$ is $\sup_{x\in C}x_i+x_{n+1}$. Hence
\begin{equation*}
\begin{aligned}
	\phi_i(C)&\ge \frac{n-1}{n}\min_{x\in C^*}x_i+\frac{1}{n}\left(\sup_{x\in C}x_i+x_{n+1}\right)\\
	&\ge\frac{n-1}{n}\min_{x\in C^*}x_i+\frac{1}{n}(y_i+y_{n+1})\\
	&=y_i+\frac{y_{n+1}}{n}\\
	&\ge \min_{x\in C^*}x_i+\frac{x_{n+1}}{n}.
\end{aligned}
\end{equation*}
Similarly, since the term corresponding to $S\subseteq N\setminus i$ is at most $\sup_{x\in C}x_i$,
\begin{equation*}
\begin{aligned}
	\phi_i(C)&\le \frac{n-1}{n}\sup_{x\in C}x_i+\frac{1}{n}\left(\sup_{x\in C}x_i+x_{n+1}\right)\\
	&\le\frac{n-1}{n}\min_{x\in C}x_i+\frac{1}{n}(z_i+z_{n+1})\\
	&=z_i+\frac{z_{n+1}}{n}\\
	&\le \sup_{x\in C}x_i+\frac{x_{n+1}}{n}.
\end{aligned}
\end{equation*}

\subsection{Proof that no other solution satisfies the axioms}
Let $\phi$ satisfy the axioms.
Let $\psi: \mathcal{C} \to \mathbb{R}^{n+1}$ be given by $\psi_i = \phi_i$ for all $i \in N$ and $\psi_{n+1} = 0$.
Note that $\psi$ may be viewed as a solution of the baseline model with $n+1$ players, and that it satisfies \cref{ax_efficiency,ax_additivity,ax_continuity}.
Then, setting 
\begin{equation*}
\begin{aligned}
	\mathcal{C}_0 &= \left\{C \in \mathcal{C} : \max_{x \in C} x_{N+1} = 0 \right\},
	\\\mathbb{S}_+^\dag &= \{x \in \mathbb{R}^{n+1}_+ : \|x\| = 1\} \setminus \left\{\mathbf{1}^{n+1}/\sqrt{n+1}\right\}, \text{and}
	\\ \mathbb{S}_0^\dag &= \left\{x \in \mathbb{R}^{n+1} : \|x\| = 1 \text{ and }x_{N+1} = 0\right\},
\end{aligned}
\end{equation*}
\cref{lemma:representation} implies that there exists a finite Borel measure $\mu$ on $\mathbb{S}_+^\dag$ and a Borel-measurable function $h : \mathbb{S}_+^\dag \to \mathbb{S}_0^\dag$ such that
\begin{equation}
	\label{eq:representation_prime}
	\psi(C) = \int \left(\sup_{x \in C} x \cdot v\right) h(v) \diff \mu(v) 
\end{equation}
for every $C \in \mathcal{C}_0$.

Let $\mathbb{B}_+^\dag(v,r) = \{w \in \mathbb{S}_+^\dag : \|w-v\| \le r\}$ for every $v\in\mathbb{S}_\dag^+$ and $r>0$. Let $\mathbb{L}_\dag$ be the set of $v \in \mathrm{supp}(\mu)$ such that
\begin{equation*}
	\lim_{r \downarrow 0} \frac{\int_{\mathbb{B}_+^\dag(v,r)} \|h(w)-h(v)\| \diff \mu(w)}{\mu(\mathbb{B}_+^\dag(v,r))} = 0.
\end{equation*}
The Lebesgue differentiation theorem ensures that $\mu$-almost every $v \in \mathbb{S}_+^\dag$ lies in $\mathbb{L}_\dag$ \parencite[corollary 2.9.9]{federer:1969}.
Let $v^\dag \in \mathbb{R}^{n+1}$ be given by $v^\dag_i = 0$ for all $i \in N$ and $v^\dag_{n+1} = 1$, and $\rho : \mathbb{S}_+^\dag \to \mathbb{R}^n$ be given by
\begin{equation*}
	\rho(v) = 
	\begin{dcases}
		\frac {v_i}{\sqrt{\sum_{i \in N} v_i^2}} & \text{if $i \in N$ and}\\
		0 & \text{if $i = n+1$.}
	\end{dcases}
\end{equation*}
\begin{lemma}
	\label{lemma:vs_prime}
	$\rho(\mathbb{L}_\dag \setminus \{v^\dag\}) \subseteq \{v^S : S \in \mathcal{N}\}$.
\end{lemma}		

\begin{proof}
	Fix $v \in \mathbb{L}_\dag \setminus \{v^\dag\}$ and, seeking a contradiction, suppose that there is no $S \in \mathcal{N}$ such that $\rho(v) = v^S$.
	Since $h(v) \in \mathbb{S}_0^\dag$, it is enough to show that $h(v) = \mathbf{0}^{n+1}$.
	To this end, choose $v$ such that $S = \{i \in N : v_i > 0\}$ is maximal; that is, such that $S \subset \{i \in N : w_i > 0\}$ for some $w \in \mathbb{L}_\dag$ only if $\rho(w) = v^T$ for some $T \in \mathcal{N}$.
	Let $\nu : \mathbb{R}^{n+1} \to \mathbb{R}^{n+1}$ be given by
	\begin{equation*}
		\nu_i(x) = 
		\begin{dcases}
			- \frac{x_{n+1}}{n} & \text{if $i \in N \setminus S$ and} \\
			x_i & \text{otherwise} \\
		\end{dcases}
	\end{equation*}
	and let
	\begin{equation*}
	\begin{aligned}
		H &= \{x \in \mathbb{R}^{n+1} : x_i = 0 \text{ for each } i \in N \setminus S\},
		\\\lambda &= \max \left\{\frac{x_{N+1}}{n} : x \in \ch\left(\nu\left(\mathbb{S}_+^\dag \cap H\right)\right)\right\}, \text{ and}
		\\C_0 &= \ch\left(\nu\left(\mathbb{S}_+^\dag \cap H\right)\right)-\lambda\mathbf{1}^{n+1},
	\end{aligned}
	\end{equation*}
	noting that $C_0$ is well-defined, since $\mathbb{S}_+^\dag \cap H$ is nonempty and bounded above, and lies in $\mathcal{C}_0$.
	Define
	\begin{equation*}
	\begin{aligned}
		\xi^\alpha(w) &= (1+\alpha)\nu(w) - \lambda \mathbf{1}^{n+1} \text{ and} \\
		C_{\alpha r} &= \ch\left(C_0 \cup \xi^\alpha\left(\mathbb{B}_+^\dag(v,r) \cap H\right)\right)
	\end{aligned}
	\end{equation*}
	for every $\alpha,r > 0$ and $w \in \mathbb{S}_+^\dag$, noting that $C_{\alpha r} \in \mathcal{C}_0$ for sufficiently small $\alpha$ and $r$.

	\emph{Step 1:}	We first show that the solutions at $C_0$ and $C_{\alpha r}$ are equal for sufficiently small $\alpha$ and $r$.

	Note that, for all $i \in N \setminus S$, 
	\begin{equation*}
	\begin{aligned}
		\phi_i(C_{\alpha r}) + \lambda &= \phi_i\left(C_{\alpha r} + \lambda\mathbf{1}^{n+1}\right) = \phi_i\left(\nu\left(H \cap \left(\mathbb{S}_+^\dag \cup \left((1+\alpha)\mathbb{B}_+^\dag(v,r)\right)\right)\right)\right)
		\\&= 0 = \phi_i\left(\nu\left(\mathbb{S}_+^\dag \cap H\right)\right) =\phi_i\left(C_{0} + \lambda\mathbf{1}^{n+1}\right) = \phi_i\left(C_0\right) + \lambda,
	\end{aligned}
	\end{equation*}
	where the first and last equalities follow from \cref{axprime_additivity,axprime_efficiency,axprime_anonymity}, the second holds since $\nu$ is linear, and the third and fourth follow from \cref{axprime_compromise}.
	By \cref{axprime_efficiency}, it is therefore enough to show that, given $i \in S$ and sufficiently small $\alpha$ and $r$, $\phi_i(C_{\alpha r}) \ge \phi_i(C_0)$.
	In turn, by \cref{axprime_mon,axprime_continuity}, it suffices to exhibit, given $i \in S$ and sufficiently small $\alpha$ and $r$, and for all $x \in C_{\alpha r}$, a $y \in C_0$ such that $y_i \le x_i$ and $y_j = x_j$ for all $j \in N \setminus i$.\footnote{To see why this suffices, choose $\{x^k\}_{k = 1}^\infty \subset C_{\alpha r}$ such that $\ch(C_0 \cup \{x^k\}_{k = 1}^\infty) = C_{\alpha r}$.
	For all $k \in \mathbb{N}$, choose $y^k \in C_0$ such that $y^k_i \le x^k_i$ and $y^k_j = x^k_j$ for each $j \in N \setminus i$, and note that $\phi_i(C_0) \le \phi_i(C_0 \cup x^1) \le \phi_i(C_0 \cup \{x^1,x^2\}) \le \dots$, by \cref{axprime_mon}.
	Then, $\phi_i(C_{\alpha r}) = \lim_{K \to \infty} \phi_i(C_0 \cup \{x^k\}_{k = 1}^K) \ge \phi_i(C_0)$, where the first equality follows from \cref{axprime_continuity}. 
	}
	To this end, fix $i \in S$ and note first that, given sufficiently small $\alpha$ and $r$ and for every $w \in \mathbb{B}_+^\dag(v,r) \cap H$, there exists $\hat w \in \mathbb{S}_+^\dag \cap H$ such that $\hat w_i \le (1+\alpha)w_i$ and $\hat w_j = (1+\alpha)w_j$ for each $j \in (N+1) \setminus i$.
	Now fix such $\alpha$ and $r$, as well as $x \in C_{\alpha r}$.
	Choose $w \in \mathbb{B}_+^\dag(v,r) \cap H$, $u \in C_0$, and $\beta \in [0,1]$ such that $x \le \beta \xi^\alpha(w) + (1-\beta)u$, and let $y=x+\beta(\hat w-(1+\alpha)w)$. Note that 
	\begin{equation*}
		y \le y -x + \beta \xi^\alpha(w) + (1-\beta)u = \beta\left(\nu(\hat w) -\lambda \mathbf{1}^{n+1}\right) + (1-\beta)u \in C_0,
	\end{equation*}
	so that $y \in C_0$, and that $y_i \le x_i$, as desired.

	\emph{Step 2:} Next, we show that the suitably normalised difference in the solutions at $C_0$ and $C_{\alpha r}$ equals $h(v)$ in the limit as $r$ and $\alpha$ vanish. 

	For all $\alpha, r > 0$, let $\Delta_{\alpha r} : \mathbb{S}^\dag_+ \to \mathbb{R}$ be given by 
	\begin{equation*}
		\Delta_{\alpha r}(w) = \frac 1 \alpha \left(\sup_{x \in C_{\alpha r}} x \cdot w - \sup_{x \in C_0} x \cdot w\right).
	\end{equation*}
	Define 
	\begin{equation*}
		W = \left\{w \in \mathbb{S}_+^\dag : w_i = 0 \text{ for some } i \in S\right\} \cup \rho^{-1}\left(\{v^T : S \subset T \subset N\}\right)
	\end{equation*} 
	and note that, for all $w \in W$, there exists $u \in \mathbb{S}_+^\dag \cap H$ such that $(\nu(u)-\nu(v)) \cdot w > 0$, since $v \ne \rho(v^S)$.
	Then, given $r$ sufficiently small, 
	\begin{equation*}
		\min_{\substack{w \in W \\ \hat v \in \mathbb{B}_+^\dag(v,r) \cap H}} \sup_{u \in \mathbb{S}_+^\dag \cap H}  (\nu(u)-\nu(\hat v)) \cdot w > 0,
	\end{equation*}
	so that $\Delta_{\alpha r}$ vanishes on $W$ for sufficiently small $\alpha$. 
	Then, $\Delta_{\alpha r}$ vanishes on $\mathbb{L}_\dag \setminus H$ for sufficiently small $\alpha$ and $r$, since $S$ is maximal.
	Therefore
	\begin{equation*}
	\begin{aligned}
		\lim_{r \downarrow 0} \lim_{\alpha \downarrow 0}\frac{\psi(C_{\alpha r})-\psi(C_0)}{\alpha \mu(\mathbb{B}_+^\dag(v,r))} 
		&= \lim_{r \downarrow 0} \frac 1 {\mu(\mathbb{B}_+^\dag(v,r))} \lim_{\alpha \downarrow 0} \int \Delta_{\alpha r}(w)  h(w) \diff\mu(w)
		\\&= \lim_{r \downarrow 0} \frac 1 {\mu(\mathbb{B}_+^\dag(v,r))} \lim_{\alpha \downarrow 0} \int_H \Delta_{\alpha r}(w)  h(w) \diff\mu(w)
		\\&= \lim_{r \downarrow 0} \frac 1 {\mu(\mathbb{B}_+^\dag(v,r))} \int_{\mathbb{B}_+^\dag(v,r) \cap H} h(w) \diff\mu(w)
		\\&= \lim_{r \downarrow 0} \frac 1 {\mu(\mathbb{B}_+^\dag(v,r))} \int_{\mathbb{B}_+^\dag(v,r)} h(w) \diff\mu(w)
		\\& = h(v),
	\end{aligned}
	\end{equation*}
	where the first equality follows from \eqref{eq:representation_prime}, since $C_0 \in \mathcal{C}_0$, and $C_{\alpha r} \in \mathcal{C}_0$ for sufficiently small $\alpha$ and $r$;
	the second holds since $\Delta_{\alpha r}$ vanishes on $\mathbb{L}_\dag \setminus H$ for sufficiently small $\alpha$ and $r$, and $\mu$ is concentrated on $\mathbb{L}_\dag$; 
	the third holds by the bounded convergence theorem, since $\Delta_{\alpha r}$ is bounded above by $\alpha$ on $H$ and equal to $\alpha$ on $\mathbb{B}_+^\dag(v,r) \cap H$, and $\lim_{\alpha \downarrow 0}\Delta_{\alpha r}(w) = 0$ for all $w \in (\mathbb{S}^\dag_+ \cap H) \setminus \mathbb{B}_+^\dag(v,r)$;
	the fourth holds since $\mathbb{B}_+^\dag(v,r) \cap \mathbb{L}_\dag \subset H$ for $r$ sufficiently small, since $S$ is maximal, and $\mu$ is concentrated on $\mathbb{L}_\dag$;
	and the fifth holds since $v \in \mathbb{L}_\dag$.

	Steps 1 and 2 together imply that $h(v)=\mathbf{0}^{n+1}$, contradicting the fact that $h(v) \in \mathbb{S}_0^\dag$.
\end{proof}

Let	$C_0=\ch(\mathbb{S}_\dag^+)-\frac{1}{\sqrt{n+1}}\mathbf{1}^{n+1}$ and define
\begin{equation*}
\begin{aligned}
	\xi^\alpha(w)&=(1+\alpha)w-\frac 1 {\sqrt{n+1}}\mathbf{1}^{n+1} \text{ and}\\
	C_{\alpha r}(v)&=\ch \left(C_0\cup \xi^\alpha\left(\mathbb{B}_\dag^+(v,r)\right)\right)
\end{aligned}
\end{equation*}
for every $\alpha,r>0$ and $v,w\in\mathbb{S}_\dag^+$. 
Note that, for all $v \in \mathbb{L}_\dag$, 
\begin{align}
	\nonumber &\lim_{r \downarrow 0} \lim_{\alpha \downarrow 0}\frac{\psi(C_{\alpha r}(v))-\psi(C_0)}{\alpha \mu(\mathbb{B}_+^\dag(v,r))} 
	\\\nonumber 
	&= \lim_{r \downarrow 0} \frac 1 {\mu(\mathbb{B}_+^\dag(v,r))}\lim_{\alpha \downarrow 0} \int \frac 1 \alpha \left(\sup_{x \in C_{\alpha r}} x \cdot w - \sup_{x \in C_0} x \cdot w \right) h(w) \diff\mu(w)
	\\\label{eq:lebesgue} &= \lim_{r \downarrow 0} \frac 1 {\mu(\mathbb{B}_+^\dag(v,r))} \int_{\mathbb{B}_+^\dag(v,r)} h(w) \diff\mu(w) = h(v),
\end{align} 
where the first equality follows from \eqref{eq:representation_prime} and the last holds since $v \in \mathbb{L}_\dag$.
Hence $h_{n+1}(v) = 0$ for every $v\in\mathbb{L}_\dag$.

We claim that $v^\dag \notin \mathbb{L}_\dag$. To see why, note that $\phi(C_{\alpha r}(v^\dag))=\phi(C_0)=\mathbf{0}^n$ by anonymity and efficiency. If $v^\dag \in \mathbb{L}_\dag$, then $h(v^\dag) = \mathbf{0}^{n+1}$ by \eqref{eq:lebesgue}, which is impossible, since $h(v^\dag) \in \mathbb{S}_0^\dag$. Hence $v^\dag \notin \mathbb{L}_\dag$. 

Note that, given $v\in \mathbb{L}_\dag$ such that $\rho(v)=v^S$ for some $S \in \mathcal{N}$, $\phi_i(C_{\alpha r}) = \phi_j(C_{\alpha r})$ for all $i,j \in N$ satisfying either $\{i,j\} \subseteq S$ or $\{i,j\} \subseteq N \setminus S$, by \cref{axprime_anonymity}.
Since $\phi(C_0) = \mathbf{0}^n$ and $h_{n+1} = 0$ on $\mathbb{L}_\dag$, \eqref{eq:lebesgue} implies that $h_i(v)$ satisfies \eqref{eq:hS} for all $i \in N$.

For each $S\in\mathcal{N}$, let $w^S\in\mathbb{S}_+^\dag$ be given by 
\begin{equation*}
	w^S_i=
	\begin{cases}
		\frac{1}{\sqrt{|S|+1}} &\text{if }i\in S\cup\{n+1\} \text{ and}\\
		0 &\text{otherwise.}
	\end{cases}
\end{equation*}
We claim that $\mathbb{L}_\dag\subseteq\{w^S:S\in\mathcal{N}\}$. Suppose first that there exists $v\in \mathbb{L}_\dag$ such that $\rho(v)=v^S$ for some $S\in \mathcal{N}$ but $v_{n+1}> w^S_{n+1}$, and seek a contradiction. Fix $\alpha\ge0$ and let $x^\alpha \in \mathbb{R}^{n+1}$ be given by
\begin{equation*}
	x^\alpha_i=
	\begin{cases}
		-1 &\text{if }i\in S,\\
		-\alpha &\text{if }i\in N\setminus S, \text{ and}\\
		|S| & \text{if }i = n+1.
	\end{cases}
\end{equation*}
Then, for $i\in S$, \eqref{eq:representation_prime} and \cref{lemma:vs_prime} yield $\phi_i(\ch\{x^\alpha,\mathbf{0}^{n+1}\})>0$ for $\alpha$ large enough, whereas \cref{axprime_compromise} implies $\phi_i(\ch\{x^\alpha,\mathbf{0}^{n+1}\})\le0$. 

Suppose now that there exists $v\in \mathbb{L}_\dag$ such that $\rho(v)=v^S$ for some $S\in \mathcal{N}$ but $v_{n+1}< w^S_{n+1}$, and seek a contradiction. Let $y=(1,\ldots,1,-n)$ and note that $\phi(\ch(\{y\})) = \mathbf{0}^n$ by \cref{axprime_anonymity,axprime_efficiency}.
Given $\alpha \ge 0$, let $z^\alpha$ be given by
\begin{equation*}
	z_i^\alpha=
	\begin{cases}
		-\frac 12 &\text{if }i\in S,\\
		-\alpha &\text{if }i\in N\setminus S, \text{ and}\\
		\frac n 2 & \text{if }i = n+1.
	\end{cases}
\end{equation*}
Note that, for $i\in S$, \eqref{eq:representation_prime} and \cref{lemma:vs_prime} yield $\phi_i(\ch\{y,z\})>\phi_i(\ch(\{y\})) = 0$ for $\alpha$ large enough, whereas \cref{axprime_compromise} implies $\phi_i(\ch\{y,z\})\le0$. Hence, $\mathbb{L}_\dag\subseteq\{w^S:S\in\mathcal{N}\}$. 

To complete the proof, it remains to determine $\mu(w^S)$ for each $S\in\mathcal{N}$. The argument is analogous to that used in the proof of \cref{theorem:solution}.

\appendix
\crefalias{section}{appendix}

\section{Other welfarist solutions}\label{app_welfarism}
The branch of the literature following \textcite{moulin:1985a} studies a non-welfarist framework: players have preferences over a set of public actions and the solution may depend not just on the set of allocations that are achievable without transfers (i.e., the bargaining set), but also on the set of public actions that yield each allocation. Among the solutions characterised only two are welfarist in the sense of depending on the bargaining set alone. The first is a classic solution in the cost-allocation literature called `equal allocation of non-separable costs' (EANS). The second is one member of the family of solutions called `equal sharing from an individual reference level'. It is the one in which each player's reference level is her maximal utility in the bargaining set, so we refer to it as `equal sharing from the ideal level' (ESIL). Both are discussed by \textcite{moulin:1985}. For a given player $i$ and bargaining set $B$, the two solutions are 
\begin{equation*}
\begin{aligned}
	\psi^{EANS}_i(B)&=\frac 1 n\left(\max_{x\in B}x_N+\sum_{j\in N}\sup_{x\in B}x_{N\setminus j}\right)-\sup_{x\in B}x_{N\setminus i} \text{ and}\\
	\psi^{ESIL}_i(B)&=\frac 1 n \left(\max_{x\in B} x_N-\sum_{j\in N}\sup_{x\in B}x_j\right)+\sup_{x\in B}x_i
\end{aligned}
\end{equation*}
respectively, where we recall from \cref{sec_model} that $N$ is the set of players, $n=|N|$, and $x_S=\sum_{i\in S}x_i$ for every $x\in B$ and $S\subseteq N$. 

Neither of these solutions satisfies the dummy property. To see why, suppose that there are three players. If $B$ is the smallest bargaining set containing $\{(0,2,0),(0,0,1)\}$, then player 1 is a dummy player but $\psi^{EANS}_1(B)=1/3$ and $\psi^{ESIL}_1(B)=-1/3$. 

\printbibliography

\end{document}